\documentclass[11pt]{amsart}
\usepackage{amsmath,amsfonts,multirow,color,amssymb}
\usepackage{graphicx, wrapfig}
\usepackage{xcolor}

\title{A no-go theorem for one-layer feedforward networks}
\author[C. Giusti]{Chad Giusti}
\address{Mathematics Department, 
University of Nebraska -- Lincoln}
\email{cgiusti2@unl.edu}
\author[V. Itskov]{Vladimir Itskov}
\address{Mathematics Department, 
University of Nebraska -- Lincoln}
\email{vladimir.itskov@math.unl.edu}
\newtheorem{thm}{Theorem}
\newtheorem{cor}[thm]{Corollary}

 \newtheorem{prop}[thm]{Proposition}
\theoremstyle{definition}

\newcommand{\C}{{\mathcal C}}
\newcommand{\R}{{\mathbb R}}

\newcommand{\nerve}{{\mathcal{N}}}

\newcommand{\Supp}{\text{supp}}
\newcommand{\refT}[1]{Theorem~\ref{T:#1}}

\newcommand{\refP}[1]{Proposition~\ref{P:#1}}

\newcommand{\od}{\stackrel{\mbox {\tiny {def}}}{=}}
\def\RR{\mathbb{R}}
\def\C{\mathcal{C}}

\begin{document}
 \maketitle 
 \begin{center}
Department of Mathematics, University of Nebraska-Lincoln\\
\indent\;  cgiusti2@unl.edu,  vladimir.itskov@math.unl.edu
\end{center}
\vspace{.25in}

\begin{abstract}
It is often hypothesized that a crucial  role for recurrent connections in the brain  is to constrain the set of possible response patterns, thereby shaping the neural code.   This implies the existence of  neural codes that cannot arise solely from feedforward processing.  We set out to find such codes in the context of one-layer feedforward networks, and identified a large class of combinatorial codes that indeed cannot be shaped by the feedforward architecture alone.  However, these codes are difficult to  distinguish  from codes that share the same sets of maximal activity patterns in the presence of noise.  When we coarsened the notion of combinatorial neural code to keep track only of maximal patterns, we found the surprising result that all such  codes can in fact be realized by one-layer feedforward networks.  This suggests that recurrent or many-layer feedforward architectures are not necessary for shaping the (coarse) combinatorial features of neural codes.  In particular, it is not possible to infer a computational role for recurrent connections from the combinatorics of neural response patterns alone.

Our proofs use mathematical tools from classical combinatorial topology, such as the nerve lemma and the existence of an inverse nerve.  An unexpected corollary of our main result is that any prescribed (finite) homotopy type can be realized by a subset of the form $\RR^n_{\geq 0} \setminus \mathcal{P}$, where $\mathcal{P}$ is a polyhedron.
 
 \end{abstract}
\newpage

\section{Introduction} 

It is often hypothesized that one of the central roles of recurrent connections in the brain  is to constrain the set of possible neural response patterns, thereby shaping the neural code   \cite{Douglas2007,Luczak2009}.  This hypothesis is appealing because it provides a concrete computational function for the prevalence of recurrent connectivity in cortical areas.  It also implies the existence of neural codes that {\it cannot} arise from the structure of feedforward connections alone.
 
We test this hypothesis by analyzing the neural codes of feedforward networks. Although it is well-known that  a feedforward network with hidden layers can approximate any continuous function \cite{Cybenko89, Hornik91}, this is not the case for one-layer networks. It is thus reasonable to expect that there exist neural  codes that cannot arise in one-layer feedforward networks, and are thus necessarily shaped by a more complex network structure.

 For simplicity, we consider combinatorial neural codes  \cite{Schneidman:2006:Nature, Osborne:2008:J-Neurosci,Schneidman2011}, which consist of binary neural activity patterns and disregard details such as exact firing rates or precise spike timing. Our first result
identifies a large class of combinatorial codes that are not realized by one-layer feedforward networks.  However, we also find that in the presence of noise these codes are  difficult to  distinguish from other codes that share the same sets of maximal activity patterns.  In order to increase robustness to noise, we coarsen the notion of combinatorial code to keep track only of maximal patterns, and again seek to find codes that cannot arise in one-layer feedforward networks.  To our surprise, there are none.  

Our main result is a ``no-go'' theorem stating that all coarse combinatorial codes can in fact be realized by one-layer feedforward networks.  
Our proof of this theorem is constructive, and uses mathematical tools from classical combinatorial topology, such as the nerve lemma and the existence of an inverse nerve.  An unexpected corollary is that any prescribed (finite) homotopy type can be realized by a subset of the form $\RR^n_{\geq 0} \setminus \mathcal{P}$, where $\mathcal{P}$ is a polyhedron.

This ``no-go'' theorem implies  that recurrent or many-layer feedforward architectures are not necessary for shaping combinatorial features of neural codes.  In particular, it is not possible to infer a computational role for recurrent connections solely from the combinatorics of neural response patterns.    
However, we also show that one-layer feedforward networks that respect Dale's law \cite{DayanAbbott} 
 produce fairly trivial neural codes, possessing  just one maximal activity pattern. Thus, one-layer feedforward networks can only produce interesting  codes if the projections from each  input  neuron are allowed to have both positive and negative weights.  This suggests that recurrent architecture or hidden layers may  be necessary to compensate for the restrictions imposed by Dale's law.

 \section{Background}
 
 A {\it combinatorial neural code} of a given population of $n$ neurons  is the collection of all possible combinations of neurons that can be simultaneously  active. More precisely, a {\it codeword} is a subset 
 $$\sigma\subset [n]\od \{1,\dots,n\}$$ 
 of neurons that are simultaneously active (i.e. fire within some small temporal window)  at some point of time, while a combinatorial {\it code} is the collection of all such subsets
$$\C \subset 2^{[n]} \od \{\text{subsets of }[n]\}.$$
 Note that this  notion of neural code describes only the set of possible response patterns of a network \cite{Schneidman:2006:Nature, Osborne:2008:J-Neurosci,Schneidman2011}, but does not include the ``dictionary'' of relationships between response patterns and network inputs.

 \begin{figure} 
    \vspace{-.15in}
    \setlength{\unitlength}{2in}
    \begin{picture}(1,0.68647504)%
      \put(0,0){\includegraphics[width=\unitlength]{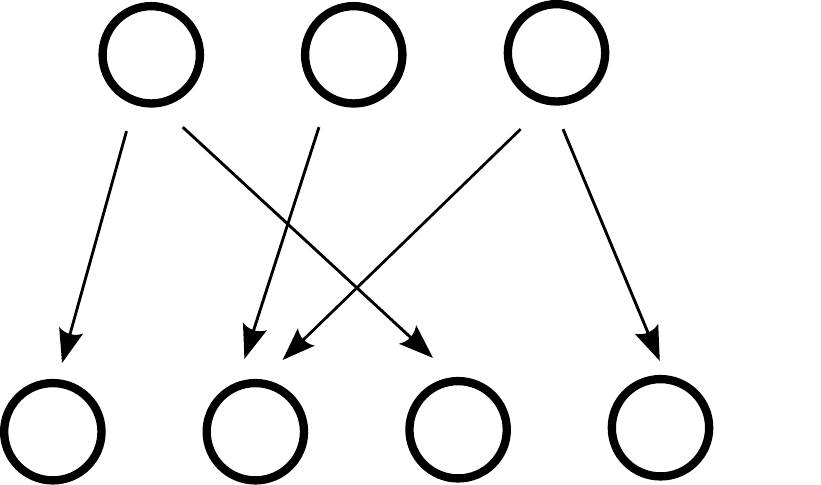}}%
      \put(0.9,0.04){\color[rgb]{0,0,0}\makebox(0,0)[lb]{$x_i$}}%
      \put(0.78,0.5){\color[rgb]{0,0,0}\makebox(0,0)[lb]{$y_a$}}%
      \put(0.77,0.26){\color[rgb]{0,0,0}\makebox(0,0)[lb]{$U_{ia}$}}%
    \end{picture}%
    \vspace{.1in}
    \caption{One-layer feedforward network}
    \label{f:network}
  \end{figure}
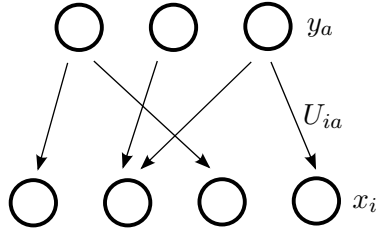 
We investigate what combinatorial codes can arise as response patterns in a {\it one-layer feedforward network} -- that is, a collection of { uncoupled}  neurons ({\it perceptrons})   \cite{ROSENBLATT1958}, driven by an input layer of neurons (Figure \ref{f:network}).  The  firing rates, $x_i(t)\geq 0$    of such neurons   can be approximated by 
the equations 
\begin{equation}\label{eq:network:rate:model}
x_i= \phi_i\left(\sum_{a=1}^m U_{ia}y_a  -\theta_i  \right), \;\qquad  i = 1,...,n.
\end{equation} 
where  $\theta_i>0$ are the neuronal thresholds,    $U_{ia}$ are the effective strengths of the feedforward connections,  $y_a(t)\geq 0$ are the firing rates of the neurons in the input layer,  and the {\it transfer functions} $\phi_i\colon{\mathbb R}\to  {\mathbb R_{\geq 0}}$ are  monotone increasing.

  For  a non-negative  firing rate vector $x\in \mathbb R^n_{\geq 0},$ we denote the subset of co-active neurons as 
\begin{equation*}\Supp(x)=\{i\in [n]\;  \vert \; x_i>0\}.
\end{equation*}

   Given a particular choice of the transfer functions $\phi_i$, the combinatorial neural code $\C(U,\theta)\subseteq 2^{[n]}$ of the network described by equation  \eqref{eq:network:rate:model} is the collection of all possible subsets of neurons that can be co-activated by nonnegative firing rate inputs $y\in \mathbb R^m_{\geq0}$:
\begin{equation*} \C(U,\theta)=\left \{\sigma=\Supp(x(y)) \,\; \vert \; 
\,  y\in \mathbb R^m_{\geq 0}    \right\}.
\end{equation*}

Our definition of $\C(U,\theta)$ captures all possible response patterns of the network, without any constraints on the network inputs.
While it may be more realistic to assume that the inputs $y$ are themselves constrained to a subset $Y \subset  \mathbb R^m_{\geq0}$, our primary interest is in combinatorial codes
that are shaped by the structure of the feedforward network, as opposed to codes whose structure is largely inherited from the previous layer.  As an extreme example, consider a ``relay'' network where the output layer is a copy of the input later, and each input neuron relays its activity to exactly one output neuron.\footnote{I.e., $U$ is the identity matrix.} Then, any code could in principle be realized via an appropriate choice of $Y$, even though the network itself plays no role in shaping the code.
Because we are interested in networks whose computational function is to shape the structure of the code, we consider 
 the opposite extreme, where input firing rates are allowed to range over all nonnnegative values, and the code $\C(U,\theta)$ is completely determined by the feedforward connections $U$ and the thresholds $\theta$.

It is well-known that  a collection of $n$ perceptrons described by equation \eqref{eq:network:rate:model} can be thought of as an arrangement of $n$ hyperplanes in the non-negative orthant $\mathbb R^m_{\geq 0},$ where $m$ is the size of the input layer. Without loss of generality, we can   assume that the monotone increasing functions $\phi_i(t)$ satisfy the conditions 
\begin{equation} \label{eq:phi} \phi_i(t)=0,  \,\,  \text{ if }  t \leq 0,\quad \text{ and } \quad \phi_i(t) >0 \,\text{ if } t>0.
\end{equation}
This condition implies that, given an input   $y\in \mathbb R^m_{\geq 0},$  
$$ x_i>0 \iff   \quad \sum_{a=1}^m U_{ia}y_a  -\theta_i >0.$$
The combinatorial code of the network, $\C(U,\theta)$, can thus be identified with the list of the regions into which the  above hyperplanes partition the positive orthant $\mathbb R^m_{\geq 0},$ as in Figure \ref{f:hyperplane}.    More precisely\footnote{Here, by convention:  $ \bigcap_{i\in\emptyset}H^+_i =\bigcap_{j\notin [n]}H^-_j=\R^m_{\geq 0}. $},  
\begin{equation}\label{eq:codedescription}
\C(U,\theta)=\left\{\sigma\subseteq [n] \,\, \vert \,\,\left( \cap_{i\in\sigma} H^+_i \right)\cap \cap_{j\notin\sigma} H^{-}_j \neq \emptyset\right\},
\end{equation} 
where
\begin{align*}
H_i^+ &=  \{y \in \R^m_{\geq 0} \;|\; (Uy)_i > \theta_i\},\quad \text{ and } \\
H_i^- &=  \{y \in \R^m_{\geq 0} \;|\; (Uy)_i \leq  \theta_i\}.\\
\end{align*}

\begin{figure}
 \setlength{\unitlength}{2in}
  \begin{picture}(1,1)%
    \put(0,0){\includegraphics[width=2in, height=2in]{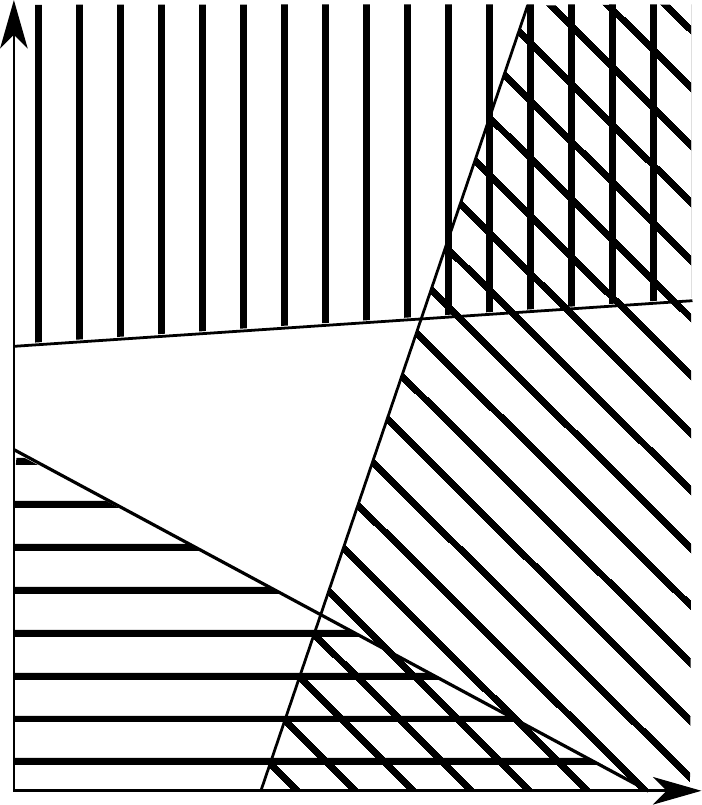}}%
    \put(0.3,0.45){\color[rgb]{0,0,0}\colorbox{white}{\tiny $\emptyset$}}%
    \put(0.21,0.75){\color[rgb]{0,0,0}\colorbox{white}{\tiny $\{1\}$}}%
    \put(0.66, 0.33){\color[rgb]{0,0,0}\colorbox{white}{\tiny $\{2\}$}}%
    \put(0.08,0.18){\color[rgb]{0,0,0}\colorbox{white}{\tiny $\{3\}$}}%
    \put(0.73,0.78){\color[rgb]{0,0,0}\colorbox{white}{\tiny $\{1,2\}$}}%
    \put(0.45,0.1){\color[rgb]{0,0,0}\colorbox{white}{\tiny $\{2,3\}$}}%
    \put(.97,-.05){\color[rgb]{0,0,0}\colorbox{white}{\tiny $y_1$}}%
    \put(-0.1,1){\color[rgb]{0,0,0}\colorbox{white}{\tiny $y_2$}}%
    
  \end{picture}%
  \caption{Codewords in $\C(U,\theta)$ correspond to regions in the positive orthant $R^m_{\geq 0}$.}
  \label{f:hyperplane}
\end{figure}
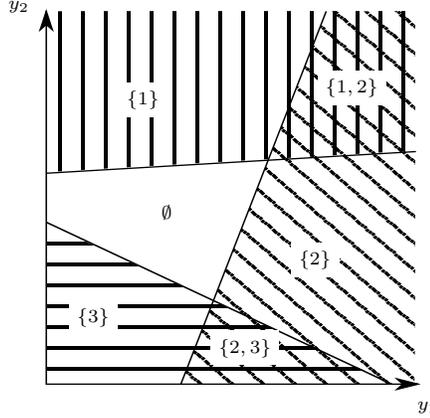

\section{Results}

We first observe that not every combinatorial neural code $\C \subset 2^{[n]}$ can be encoded by a one-layer feedforward network.   A  code $\C $    is called {\it convex} \cite{neural-ring}  if there exists a collection of convex subsets $\{V_i\}_{i=1}^{n}$ of  $X\subset \mathbb R ^d$ 
such that\footnote{Here, by convention: $ \bigcap_{i\in\emptyset}V_i= \bigcap_{j\notin [n]}(X\setminus V_j)=X.$}  
\begin{equation} \label{eq:convexcode}
\C= 
\{\sigma\subseteq [n] \,\, \vert \,\, ( \cap_{i\in\sigma} V_i  )\cap \cap_{j\notin\sigma} \left(X\setminus V_j\right)  \neq \emptyset \}.
\end{equation}
Note that choosing  $V_i=H_i^+$ and $X=\R^m_{\geq 0}$ yields  eq. \eqref{eq:codedescription},  thus any feedforward code $\C(U,\theta)$    is convex.  Therefore  non-convex codes  can not be encoded by  a one-layer feedforward network. 
 
Perhaps the simplest example of a non-convex code is (\cite{neural-ring}, code B3)
\begin{equation} \label{eq:simpleexample} \C = 2^{[3]}  \setminus  \{ \{1,2,3\},  \{2,3\},\{1\}\}= \{ \{2\},\{3\}, \{1,2\},  \{1,3\}\}. 
\end{equation}
We  generalize this example by considering combinatorial codes that contain subcodes that are obstructions to convexity.  A {\it subcode} of a combinatorial code $\C\subset 2^{[n]}$ is a collection of patterns obtained by restricting codewords in $\C$ to neurons inside a given subset $\sigma\subset [n]$:
\begin{equation*}\C(\sigma) \od \{\nu=\sigma\cap \mu\,\,\vert \,\, \mu\in \C\}\subseteq 2^{\sigma }.
\end{equation*}
It is easy to see that any subcode of a convex code is convex. Thus, possessing a {\it non-convex}  subcode is an obstruction to being realizable as a feedforward code $\C(U,\theta).$

The following is a  generalization of the example \eqref{eq:simpleexample} and is a corollary of the more general \refT{convex}, which we prove in Section~\ref{section:convexcodes}.
\begin{thm}\label{T:subcode} Assume a combinatorial code $\C$ posesses a subcode $\C(\sigma)\subset 2^\sigma$ that satisfies the following conditions:
\begin{enumerate}
\item[(i)]  $\sigma\notin \C(\sigma),$
\item[(ii)]  there exists a non-empty subset $\nu\subset \sigma$  with $\vert \nu\vert \leq \vert \sigma\vert-2$ such that $\nu\notin \C(\sigma),$ 
\item[(iii)]  for every subset $\sigma^\prime\subset \sigma$ with $\vert \sigma^\prime\vert=\vert \sigma \vert -1$, 
 if $\sigma^\prime\supset \nu$ then $\sigma^\prime \in \C(\sigma),$ 

\end{enumerate}
then the code $\C$ is non-convex and thus can not be realized by a one-layer feedforward network.
\end{thm} 

While the subcodes in \refT{subcode} provide a large repertoire of  codes that can be used to rule out a feedforward code,   these codes are  not robust to noise.  
Specifically, a non-convex code may differ from a convex code by as little as one neuron's participation in a single codeword.
These codes are thus nearly impossible  to distinguish from other codes that share the same set of maximal patterns in the presence of noise.  

We say that a codeword $\sigma\in \C$ is a {\it maximal pattern} of the code $\C$ if it is not contained in any larger pattern, and denote the set of all maximal patterns of $\C$ as 
$$\max (\C)=\{\sigma\in \C \;\vert \; \sigma^\prime \supsetneq \sigma \Rightarrow \sigma^\prime \notin \C \}.$$ 
We refer to $\max(\C)$ as a {\it coarse combinatorial code}. 

The coarse code $\max(\C)$ is more robust than the full combinatorial code $\C$ to ``missing'' a small number of spikes.  Experimentally observed neural activity is often  {\it sparse} \cite{Hromadka:2008,Barth:2012}, i.e.  the number of  co-active neurons in a codeword $\sigma$ is bounded as $ \vert \sigma\vert \leq s n$, with the fraction  of active neurons often as low as $s<0.1$.  Two different sparse combinatorial codes are thus likely to possess distinct sets of maximal patterns. 

We now investigate the possibility of ruling out a one-layer feedforward network based on the coarse combinatorial code $\max(C)$.  
Our central result  is that this cannot be done.  Surprisingly, one-layer  feedforward networks may encode any prescribed set of maximal patterns. 

  \begin{thm}[No-go Theorem] \label{T:ff_theorem}  For every collection $\mathcal M \subset 2^{[n]}$ of maximal patterns  there exists a one-layer  feedforward network of perceptrons \eqref{eq:network:rate:model}   with  $$\max(\C(U,\theta)) = \mathcal M.$$
\end{thm}

\noindent In particular, it is not possible to infer a computational function for recurrent connections, or hidden layers, in a network from observations of the coarse combinatorial code alone. 
 The proof of Theorem~\ref{T:ff_theorem} (see Section~\ref{sec:no-go-proof}) is constructive, and uses tools from combinatorial topology.  An unexpected corollary is the following topological fact:

\begin{cor}
\label{C:topological}
For any finite abstract simplicial complex\footnote{See Section \ref{s:proofs}.} $\Delta$, there exists a subset of the form $\RR^m_{\geq 0} \setminus \mathcal{P}$, for $\mathcal{P}$ a polyhedron, which is homotopy equivalent to $\Delta$.
\end{cor}

This is perhaps counterintuitive because if the polyhedron is fully-contained within $\RR^m_{>0}$, the  complement has the homotopy type of a sphere.  The richness in the topology emerges from the intersection of the polyhedron with the boundary of the positive orthant. 
  
  \bigskip 
An important caveat for interpreting our ``no-go'' theorem is that neurons in the brain typically possess a strong constraint called {\it Dale's law} \cite{DayanAbbott}. Dale's law states that neurons either have purely excitatory or purely inhibitory synapses onto other neurons\footnote{One notable exception is the gap junctions.}. More formally, a one-layer  feedforward network  \eqref{eq:network:rate:model}  {\em respects} Dale's Law if there is a partition of the columns of the synaptic matrix  into two families: $U=[U_{+}\mid U_{-}]$ so that all entries of $U_+$ are non-negative, and those of $U_-$ are non-positive.
It turns out that one-layer  networks that respect Dale's law are capable of producing only an extremely restricted class of coarse combinatorial codes.

\begin{prop} \label{P:contractible} Suppose that the one-layer feedforward network \eqref{eq:network:rate:model} respects Dale's Law. Then the combinatorial code of this network has exactly one maximal pattern  $\sigma_{\max}\in \C(U, \theta).$  I.e.,
$$ \max(\C(U, \theta))=\{\sigma_{\max}\}.$$
\end{prop}
\noindent  The proof is given in Section \ref{sec:DaleProofs}. This property, in particular, excludes most known sparse neural codes, such as place field codes or orientation tuning codes. 

It is worth noting, however, that  a one-layer feedforward network without the Dale's law constraint can be thought of as an approximation of a two-layer feedforward network that respects Dale's law.  For example, assuming  $\phi_i(y)=[y]_+\od \max(0,y),$ that  the first layer synapses $U^1$ are excitatory while the second layer synapes $U^2$ respect Dale's law, and moreover that the first layer  has zero thresholds,  one obtains the following equation for the firing rates $x$ of the output neurons: $$ x=\left[U^2\left[U^1y\right]_+-\theta\right]_+= \left[U^2 U^1y -\theta\right]_+= \left[Uy -\theta\right]_+.$$  Observe that the  resulting matrix $U=U^2U^1$ no longer has to respect Dale's law, even though the component matrices $U^1$ and $U^2$ do so. Thus two-layer feedforward networks that obey Dale's law  are capable of encoding any prescribed simplicial complex.  I.e. in the presence of a  layer of inhibitory  neurons a feedforward network can encode topologically interesting stimuli.

\section{Conclusions}
 
 Motivated by the idea that recurrent or many-layer feedforward architecture may be necessary to shape the neural code, we set out to find combinatorial codes that {\it cannot} arise from one-layer feedforward networks.  Although we found a large class of such examples, the non-convex codes, 
we also observed that in the presence of noise they would be virtually indistinguishable from other codes having the same maximal patterns.  When we considered coarse combinatorial codes, which keep track of only the maximal patterns, we found that there do not exist any codes that cannot be realized by a one-layer feedforward network. Our results suggest that recurrent connections, or multiple layers, are not necessary for shaping the neural code.

\section{Proofs of  the main results} \label{s:proofs}

In order to understand the coarse neural codes $\max (\C)$ it is convenient to consider the maximal possible code with the same $\max( \C).$ 
This can be thought of as a ``completion'' of the code $\C$ obtained by  adding all the   subsets  of $\max(\C) $;   this results in  a new code  $\Delta(\C)\supseteq \C,$
\begin{equation} \Delta(\C)=\{ \nu\subseteq \sigma\;\vert \;\sigma \in \C\}=\{\nu \subseteq \sigma\;\vert \;\sigma\in \max (\C)\}.
\end{equation}
This collection of sets is closed under inclusion, i.e. $\nu\subset \sigma\in \C$ implies that $\nu\in \C.$ 
A collection of sets with this property  is called  an {\it abstract simplicial complex}. 

\subsection{Convex codes}\label{section:convexcodes}

It is well known that  every abstract simplicial complex, i.e. a code that satisfies $\C=\Delta(\C),$   is  a convex code \cite{Wegner1967,Tancer2013}.  The following result shows that convex codes that are {\it not} simplicial complexes also have strong restrictions.

\begin{thm}\label{T:convex}Assume that $\C\subset 2^{[n]}$ is a convex code of the form \eqref{eq:convexcode} and that a codeword $\nu\in \Delta(\C)$ violates the simplicial complex property, i.e. $\nu\notin\C.$ Then, the localized complex 
\begin{equation}\label{eq:localizeddelta} \Delta(\C)_{\vert \nu }\od\left \{ \tau \subseteq  [n]\setminus \nu  \; \vert \; \left(\tau\cup \nu \right)\in \Delta(\C)\right\}\subseteq 2^{ [n]\setminus \nu }
\end{equation}
is contractible. 
\end{thm} 
\proof Denote by $V_\nu=\cap_{j\in \nu} V_i.$ It is easy to see that $ \Delta(\C)_{\vert \nu }$  is the nerve of the cover of $V_\nu$ by the convex sets $\tilde V_j \od V_j\cap V_\nu.$  Moreover, since $\nu\notin \C,$ the sets $\tilde V_j $ cover $V_\nu,$ i.e. $V_\nu =\cup_{j\notin \nu} \tilde V_j.$ Therefore we can use the nerve lemma (see e.g. \cite{Hatcher}, Corollary 4G.3 p. 460) and conclude that the simplicial complex  $ \Delta(\C)_{\vert \nu }$ is homotopy equivalent to the set  $V_\nu$ and thus contractible.
\qed

\medskip 
\noindent We now give the proof of \refT{subcode} as a corollary of  \refT{convex}.
\begin{proof}[Proof of \refT{subcode}]
 Given the conditions (i)-(iii) of \refT{subcode}, it is easy to see  that the maximal patterns   
 of  the localization  $\Delta(\C(\sigma))_{\vert \nu}$ 
 can be described as 
\begin{equation*}
\max(\Delta(\C(\sigma))_{\vert \nu}) =\left \{\tau \subset \left(\sigma \setminus \nu\right)\; \vert\;    \;\vert\tau\vert = \vert\sigma\vert - \vert \nu\vert - 1 \right\}.
\end{equation*} 
Therefore the simplicial complex $\Delta(\C(\sigma))_{\vert \nu}$  can be identified with the boundary of the $(\vert \sigma \vert - \vert \nu \vert -1)$-dimensional simplex $\sigma \setminus \nu$, and thus is not contractible. Therefore $\C(\sigma)$ is not a convex  code and thus neither  is $\C.$
\end{proof}

\subsection{Proof of the ``no-go''  Theorem \ref{T:ff_theorem}.} \label{sec:no-go-proof}
We prove the no-go theorem in  two steps. First we find an explicit construction of a feedforward network that encodes the  {\it nerve}  $\nerve(\Delta)$ of any  simplicial complex $\Delta.$ We then use a result by  Branko Gr{\"u}nbaum   \cite{Grunbaum} showing the existence of an inverse nerve.

Recall that a face $\alpha\in \Delta$ of a simplicial complex is {\it maximal} if it is not contained in any larger face  of $\Delta .$   
 Let $\Delta\subset 2^{[m]}$ be an abstract simplicial complex on $m$ vertices and  $ \{\alpha_1, \dots, \alpha_n\}=\max(\Delta)$  be the maximal sets of $\Delta $.    The {\it nerve} of $\Delta $ is another abstract simplicial complex $\nerve(\Delta )\subseteq 2^{[n]}$ on $n$ vertices, such that for any nonempty $\nu\in [n]$ 
\begin{equation*} \nu\in \nerve(\Delta ) \iff \bigcap\limits_{i\in \nu} \alpha_i \neq \emptyset.
\end{equation*}

\begin{prop}\label{P:construction} Let $\Delta \subseteq 2^{[m]}$ be a simplicial complex with   $m$ vertices 
and  $n$ maximal faces $ \{\alpha_1, \dots, \alpha_n\}.$  Then there exists a strictly feedforward network $(U,\theta)$ such that 
$$ \Delta(\C(U,\theta))=\nerve(\Delta). $$
Moreover, the network parameters $U$ and $\theta$ for  such a network   can be constructed with 
\begin{equation}\label{eq:explicitconstruction}
U_{i,a} = \begin{cases} -n & \text{if } a \not\in \alpha_i,\\
1 & \text{if } a \in \alpha_i,
\end{cases}
\qquad \text{ and } \,\,
 \theta_i= \frac12. 
\end{equation}
\end{prop}

\begin{proof} Given the choice   \eqref{eq:explicitconstruction}, the  regions $  H^+_i \subset \mathbb R ^m _{\geq 0}$ can be then described  by the inequalities $y_a \geq 0$ and 
\begin{equation}\label{eq:Ha}  \sum_{a\in \alpha_i}y_a-n \sum_{a\notin \alpha_i}y_a>\frac12.
\end{equation}
Note that equation \eqref{eq:codedescription} implies that 
$$ \Delta(\C(U,\theta))=\{\sigma\subseteq [n] \,\, \vert \,\,\bigcap_{i\in\sigma} H^+_i   \neq \emptyset\}\cup \emptyset.
$$
Thus,  to show that  $\Delta(\C(U,\theta))\supseteq \nerve(\Delta )$ we need to prove that for any non-empty 
 $ \sigma\in {\mathcal N}(\Delta),$ 
 \begin{equation}\label{eq:nonempty}
 \bigcap_{i\in\sigma} H^+_i \neq \emptyset.
\end{equation}
 For any non-empty $\alpha\subseteq [m]$ we define a vector $y^\alpha\in  \R^m_{\geq0}$ as 
 \begin{equation*} y^\alpha_a\stackrel{\operatorname{def}}{=} \begin{cases} 1 & \text{if } a  \in \alpha ,\\
0 & \text{if } a \notin \alpha.
\end{cases}
 \end{equation*} 
Let $\sigma \subseteq  [n]$ and define $\cap_{i\in\sigma}\alpha_i\od \alpha(\sigma)\subseteq [m].$ By the definition of the nerve, if  a nonempty  $\sigma \in \nerve(\Delta ),$ then  $\alpha(\sigma) \neq \emptyset$ and plugging in $y^{\alpha(\sigma)},$ the inequality \eqref{eq:Ha} becomes $\vert \alpha_i\vert>\frac12.$ Thus  $y^{\alpha(\sigma)}\in  \cap_{i\in\sigma} H^+_i $ and  the property  \eqref{eq:nonempty} holds.

\medskip
\noindent  In order to show that  $\Delta(\C(U,\theta))\subseteq \nerve(\Delta )$   we  need to prove  that  
 \begin{equation}\label{eq:empty}
 \text{for any } \tau\notin  {\mathcal N}(\Delta ),\qquad 
 \bigcap_{i\in\tau} H^+_i    = \emptyset.
\end{equation}
Assume the converse, i.e. for some $\tau\notin \nerve (\Delta)$ there  exists  $y\in  \bigcap_{i\in\tau} H^+_i.$
From the definition of the nerve we obtain  $ \alpha(\tau)=\emptyset$ and thus  for every $a\in [m]$ 
\begin{equation}\label{eq:tau} \left\vert \{ j\in\tau\;\vert\; a\in \alpha_j\}\right\vert \leq \vert\tau\vert -1,\quad \text{and} \quad  
                               \left\vert \{ j\in\tau\;\vert\; a\notin \alpha_j\}\right\vert \geq  1.
\end{equation}
Summing  the  inequalities  \eqref{eq:Ha} over all $j\in \tau$  we obtain 
\begin{equation*}\sum_{j\in\tau}\sum_{a\in\alpha_j}y_a-n \sum_{j\in\tau}\sum_{a\notin\alpha_j}y_a>\frac{\vert\tau\vert}{2},
\end{equation*}
 thus using \eqref{eq:tau} we obtain 
\begin{equation*}
\sum_{a=1}^m y_a \left(\vert\tau\vert-1-n\right) \geq \sum_{a=1}^m y_a\left( \left\vert \{ j\in\tau\;\vert\; a\in \alpha_j\}\right\vert -n\vert    \{ j\in\tau\;\vert\; a\notin \alpha_j\}    \vert\right)  
> \frac{\vert\tau\vert}{2}>0.
\end{equation*}
Taking into account  that  $y_a\geq 0$ for all $a,$  we obtain  that $\vert\tau\vert\geq n+1,$ which is a contradiction, thus the condition \eqref{eq:empty} holds.
\end{proof}

\medskip 

\noindent To finish the proof of Theorem \ref{T:ff_theorem} we use  the following classical result  of   Gr{\"u}nbaum.  (For completeness, we also give  Gr{\"u}nbaum's  explicit construction of the inverse nerve in the Appendix.) 
\begin{thm}[\cite{{Grunbaum}}] For every abstract simplicial complex $\Delta$ there exists a simplicial complex $ \tilde \Delta,$ such that $\nerve(\tilde\Delta)=\Delta.$
\end{thm}

\begin{proof}[Proof of Theorem \ref{T:ff_theorem}.] Given a set $\mathcal M\subset 2^{[n]}$ of maximal patterns, choose an abstract simplicial complex  $\tilde\Delta$ such that $\nerve(\tilde\Delta)=\Delta(\mathcal M).$ By Proposition \ref{P:construction}, the prescription \eqref{eq:explicitconstruction} yields a one-layer network network such that 
$$\Delta(\C(U,\theta))=\nerve(\tilde \Delta)=\Delta(\mathcal M) \quad \implies \quad\max (\C(U,\theta))=\mathcal M.$$
\end{proof}
\begin{proof}[Proof of Corollary \ref{C:topological}]
Given an abstract simplicial complex $\Delta \subseteq 2^{[n]}$, consider  a  one-layer neural network $(U, \theta)$  from   \refT{ff_theorem}, so that  
$\max(\C(U,\theta)) = \max (\Delta).$ The codewords in $\C(U, \theta)$ are identified with non-empty polyhedra  in the positive orthant as in equation \eqref{eq:codedescription}, with the empty codeword $\sigma=\emptyset$ identified with the polyhedron $\mathcal{P}\od \bigcap_{i\in[n]} H_i^- .$ 
 
Now, observe that the region $Y=\RR^m_{\geq 0} \setminus \mathcal{P} $ is covered by the convex sets $ \{H_i^+ \}_{i \in [n]}$, and for any non-empty $\sigma \subseteq [n]$ the intersection $\cap_{i \in \sigma} H_i^+$ is non-empty  if and only if $\sigma \in \Delta(\C(U, \theta))=\Delta$. Thus $\Delta$  is the nerve of the cover of $Y$ by the convex sets $\{H_i^+\}$.Therefore,   by the nerve lemma\footnote{See e.g. \cite{Hatcher}, Corollary 4G.3 p. 460.}, $Y$ is homotopy equivalent to $\Delta$.
\end{proof}

\subsection{Networks constrained by Dale's Law}\label{sec:DaleProofs}
\begin{proof}[Proof of \refP{contractible}]

Suppose   each neuron in the input layer is either excitatory or inhibitory, i.e.  one can reorder the input neurons so that $U = [U_+ \mid U_-]$. Denote by  $\sigma_{\max}\subseteq\{1,\dots, n\}$  the set of all neurons in the output layer that are  either ``on'' in the absence of external drive  or  receive at least one excitatory connection: 
\begin{equation*}
\sigma_{\max}= \{ i\;\vert -\theta_i>0\}\cup \{i \;\vert  \, \text{ there exists } j \text{ with }  U_{ij}>0   \}.\, 
\end{equation*}
Because $U=[U_+ \mid U_-],$  setting the firing rates of the excitatory neurons sufficiently high\footnote{
Setting  $ y_j>\max_{\{i: U_{ij}>0\}} \left\{\frac{\vert \theta_i\vert }{U_{ij}} \right\}$ is sufficient.
} and the firing rates of the inhibitory neurons to zero yields a firing rate vector $x\in \mathbb R^n_{\geq0}$ with support  $\sigma_{\max}\in \C(U,\theta).$ Moreover,  any element of the code $\C(U,\theta)$ must be a subset of $\sigma_{\max}.$
\end{proof}

\bigskip
\noindent {\bf Acknowledgements:} This work was supported by NSF DMS-1122519. V.I. is thankful to Carly Klivans for pointing to the reference \cite{{Grunbaum}}; the authors are also grateful to Carina Curto for comments on  the manuscript. 
\bigskip 

\appendix
\section{Inverse of the nerve functor}

For completeness, we include the construction due to Branko Gr{\"u}nbaum   \cite{Grunbaum} of a simplicial complex $\tilde{\Delta}$ in the collection   $\mathcal{N}^{-1}(\Delta)$ for any abstract simplicial complex $\Delta$. 

 \label{Co:nerve_inv} Suppose $\Delta \subseteq 2^{[n]}$ is an abstract simplicial complex with maximal faces $\max(\Delta) = \{\alpha_1, \dots, \alpha_k\}$. Write $I(\Delta)$ for collection of vertices of $\Delta$ which can be written as an intersection of its maximal faces, $$I(\Delta) = \{v \in 2^{[n]}\; |\;  \exists \;\sigma  \subseteq [k] \text{ with } \{v\} = \bigcap_{i\in \sigma } \alpha_i\}.$$
 Now, choose the vertices of $\tilde{\Delta}$ to be the maximal faces of $\Delta$ along with vertices of $\Delta$ which do not appear as intersections of maximal faces, $$V = \max(\Delta) \cup (2^{[n]} \setminus I(\Delta)).$$  The maximal faces $\max(\tilde{\Delta}) = \{\beta_i\}_{i \in [n]}$, are chosen in one-to-one correspondence with the vertices of $\Delta$, with face $\beta_i$ supported on elements of $V$ which contain vertex $i$, 
$$
\beta_i = \begin{cases}\{\alpha_j \in \max(\Delta)\; |\; i \in \alpha_j\}& \text{ if }i \in I(\Delta)\\
\{\alpha_j \in \max(\Delta) \;|\; i \in \alpha_j\} \cup \{i\}&\text{ if }i \not\in I(\Delta) \\
\end{cases}.
$$
By construction, the only maximal intersections of the elements of $\max(\nerve(\tilde{\Delta}))$ are those elements of $V$ corresponding to elements of $\max(\Delta)$, thus $\nerve(\tilde\Delta)=\Delta$  and  $\tilde{\Delta} \in \nerve^{-1}(\Delta)$.
 
\medskip 
\noindent {\it Example.}  \label{Ex:inv}Let $\Delta\subseteq 2^{[6]}$ have maximal faces $\max(\Delta) = \{\{1,2,3\}$, $\{2,3,4\}$ and $\{2, 5, 6\}\}$. The set $I(\Delta)$ consists of the single vertex $2$, which is the only vertex which appears as an intersection of maximal faces.
The complex $\tilde{\Delta}$ thus has vertices $\{\{1, 2, 3\}, \{2,3,4\}, \{2,5,6\}, 1, 3, 4, 5, 6\}$ and maximal faces $\{\{1, \{1,2,3\}\}$, $\{\{1,2,3\}$, $\{2,3,4\}$, $\{2,5,6\}\}$, $\{3,\{1,2,3\}$, $\{2,3,4\}\}$, $\{4, \{2,3,4\}\}$, $\{5, \{2,5,6\}\}$, $\{6, \{2,5,6\}\}$.

\bibliographystyle{plain}
\bibliography{refs}
\end{document}